\documentclass[a4paper,11pt]{article}

\usepackage[english]{babel}
\usepackage[a4paper]{geometry}
\usepackage{enumerate}
\usepackage{amsmath}
\usepackage{amsthm}
\usepackage{amssymb}
\usepackage{hyperref}
\usepackage{amssymb}
\usepackage{color}
\usepackage[utf8]{inputenc}
\usepackage{graphicx}
\usepackage{tikz}
\usepackage{tikz-cd}
\usepackage{verbatim} %needed for \begin{comment}
	\usepackage{mathtools} %needed for \coloneqq
	\usepackage{mathabx} %needed for big times
	\usepackage{hyperref}
	\usepackage{comment}
	\usetikzlibrary{decorations.markings}
	\usepackage{bbm} %needed for indicator function
	\usepackage{xfrac} %nice looking quotient rings \sfrac

	\usepackage{pgfplots}              %those three are used to get the picture of the sublevel sets
	\pgfplotsset{compat=newest}
	\usepgfplotslibrary{fillbetween}
	
	\newtheorem{theorem}{Theorem}[section]  % use thm for %Theorems to keep numbering consistent
	\newtheorem{cor}[theorem]{Corollary}
	\newtheorem{prop}[theorem]{Proposition}
	\newtheorem{lemma}[theorem]{Lemma}

	\numberwithin{equation}{section}
	\theoremstyle{definition}

\newcommand{\caN}{{\mathcal N}}

\newcommand{\caQ}{{\mathcal Q}}

\newcommand{\caV}{{\mathcal V}}

\newcommand{\bbN}{{\mathbb N}}

\newcommand{\bbR}{{\mathbb R}}

\newcommand{\cK}{C_{\mathrm K}}
\newcommand{\cD}{C_{\mathrm D}}
\newcommand{\cH}{C_{\mathrm H}}

\newcommand{\ess}{\underset{Q}{\sup}}
\newcommand{\esi}{\underset{Q}{\inf}}

\begin{document}
		\title{Counting eigenvalues of Schrödinger operators using the landscape function}

		\author{Sven Bachmann$^1$, Richard Froese$^1$, Severin Schraven$^1$ \\
			\\
			$^1$ Department of Mathematics, The University of British Columbia, \\
			Vancouver, BC V6T 1Z2, Canada}
		
		\maketitle
		
		\begin{abstract}
We prove an upper and a lower bound on the rank of the spectral projections of the Schr\"odinger operator $-\Delta + V$ in terms of the volume of the sublevel sets of an effective potential $\frac{1}{u}$. Here, $u$ is the `landscape function' of~\cite{DFM}, namely a solution of $(-\Delta + V)u = 1$ in $\bbR^d$. We prove the result for non-negative potentials satisfying a Kato-type and a doubling condition, in all spatial dimensions, in infinite volume, and show that no coarse graining is required. Our result yields in particular a necessary and sufficient condition for discreteness of the spectrum. In the case of polynomial potentials, we prove that the spectrum is discrete if and only if no directional derivative vanishes identically.
		\end{abstract}

%%%%%%%%%%%%%%%%%%%%%%%%%%%%%%%%%%%%%%%%%%%%%%%%%%%%%%%%%%%%%%%		
\section{Introduction}

In a celebrated body of work, Fefferman and Phong~\cite{FP} carried out an extensive analysis of the spectrum of self-adjoint differential operators based on the uncertainty principle, namely the fact that there is lower bound on the localization of the Fourier transform of a function that is well localized in space. Among the far reaching consequences of this old observation, they show that the number of eigenvalues $E_j$ of a Schr\"odinger operator with positive polynomial potential $V$, below an energy $\mu$, is equivalent to a coarse-grained notion of the volume of the sublevel sets of the potential at $\mu$. Precisely, they count the number of boxes of side length of order $\mu^{-1/2}$ inside which $V$ is less than or equal to $\mu$. This coarse graining is shown to be a necessary feature that arises from the uncertainty principle: for a test function to fit into a very narrow box, its kinetic energy must be large.

The Fefferman-Phong result is a wide generalization of the classical Weyl law, which is an asymptotic result for the number of eigenvalues of $-\Delta - \lambda V$ as $\lambda\to\infty$, see~\cite{W1, W2, W3, W4}, and \cite{Iv} for many extensions and variations. It is also closely related to the Lieb-Thirring inequalities \cite{LT} (see also \cite{FLW} for a recent overview)
\begin{equation*}
\sum_{j}E_j^\gamma \leq L_{\gamma, d}\int_{\bbR^d}V(x)^{\gamma + d/2}dx.
\end{equation*}
The case $\gamma = 0$ of interest to us in this work was obtained in dimensions $d\geq 3$ independently and by very different techniques by Cwickl~\cite{Cw}, Lieb~\cite{Li}, and Rozenblum~\cite{Ro, Ro1}, and we shall refer to it as the CLR inequality. It is well-known that the CLR inequality cannot hold in complete generality for dimensions $d=1,2$. A general result in one dimension is the classical Calogero inequality \cite{Ca}. It can further be used to obtain CLR-type bounds in some two-dimensional cases, see e.g.~\cite{NL, La1, Sha, GN, LRS}. For more general kinetic energies and simplified proofs of the CLR inequality see \cite{F, HHRV, HKRV} and the references therein.

Shen generalized the results of Fefferman-Phong to potentials in some reverse H\"older class and also in the presence of a magnetic field \cite{She2, She3}. For this he used previously established $L^p$ estimates for such Schrödinger operators \cite{She1} (in the case of nonnegative polynomials, see also Smith \cite{Sm} and Zhong \cite{Zh}). These results rely on estimates on the Green's function, see Davey-Hill-Mayboroda \cite{DHM}, as well as Mayboroda-Poggi \cite{MP} for optimal kernel estimates for more general elliptic operators and in the presence of a magnetic field, in dimensions $d\geq 3$. Poggi \cite{P} considered potentials of Kato-type defined in more general domains. In dimension $d=2$ Christ \cite{Ch} obtained kernel estimates under stronger assumptions. We point out that Otelbaev obtained similar two-sided estimates earlier \cite{O1, O2}, see also~\cite{MN, MLBN} for related bounds.

The Fefferman-Phong approach was also extended by David-Filoche-Mayboroda \cite{DFM} who introduced a new technique which is central to the present work. While the bounds of Fefferman-Phong only depend on the dimension and the degree of polynomial, the dependence on the potential in~DFM is more subtle and relies on the so-called landscape function. One considers $-\Delta + V$  on finite boxes $\Lambda_L$ of side length $L$ and define the \emph{landscape function} as the solution of 
		\begin{equation}\label{landscape eq}
			(-\Delta +V) u_L =1
		\end{equation}
		with suitable boundary conditions. A formal computation shows that the operator $-\Delta + V$ on~$L^2(\Lambda_L,dx)$ is unitarily equivalent to the elliptic operator
		\begin{align*} %\label{eq:conjugatedOperator}
			-\frac{1}{u_L^2}\mathrm{div}\,u_L^2 \nabla + \frac{1}{u_L}
		\end{align*}
		on $L^2(\Lambda_L, u_L^2(x) dx)$. We shall henceforth call $1/u_L$ the \emph{effective potential}. The landscape function encodes (via the effective potential) a substantial part of the spectral information of the original Schr\"odinger operator. Among the many rigorous results using the effective potential~\cite{ADFM, St1, BLG, FMT, DFM, WZ, AFMWZ, St2, CWZ}, we are interested in the following, see~\cite{DFM}: If $\mathcal{N}^V_L(\mu)$ is the number of eigenvalues of $-\Delta + V$, counted with multiplicity, on the box with sidelength $L$ (with periodic boundary conditions) that are smaller than $\mu$, then there are constants $c,C>0$ such that
		\begin{align*} %\label{coarse bounds} 
			N(c\mu,L) \leq \mathcal{N}^V_L(\mu) \leq N(C\mu, L),
		\end{align*}
where $N(\mu, L)$ is the coarsed-grained volume corresponding to the effective potential, namely it is the number of boxes of sidelength $\mu^{-1/2}$ in which $1/u_L \leq \mu$. While $C$ only depends on the dimension, the constant $c$ depends on the oscillation of~$u_L$.

In this paper, we consider the setting of \cite{She4} where the potentials are Kato-class and satisfy a doubling condition (the precise assumptions are (\ref{assump:Kato},\ref{assump:doubling})). We first study the existence of the landscape function in the whole space $\bbR^d$, namely the existence and positivity of solutions of~(\ref{landscape eq}). This purely PDE question is set in an apriori inconvenient space since, unlike in the case of finite boxes considered in~\cite{DFM}, the right hand side belongs to no $L^p$ space but for $p=\infty$. Not surprisingly, this can be addressed by considering a sequence of compactly supported functions converging pointwise to $1$. A similar approach was used by Poggi \cite[Theorem 1.18]{P} in dimension $d\geq 3$. With the landscape function and therefore the effective potential $\frac{1}{u}$ in hand, we turn to the problem of the counting of eigenvalues. We extend the DFM result to the infinite volume setting and without coarse graining. We show that the effective potential is confining if and only if the spectrum of the Schr\"odinger operator is discrete, in which case the measure of its sublevel sets is finite. We then prove that this measure controls the eigenvalue counting function, namely
\begin{equation*}
(c\mu)^{d/2} \caV(c\mu)\leq \caN^V(\mu)\leq (C\mu)^{d/2}\caV(C\mu),
\end{equation*}
where 
\begin{equation*}
\caV(\mu) = \int_{\left\{x\in\bbR^d:\frac{1}{u(x)}\leq\mu\right\}}dx.
\end{equation*}
Crucially, this CLR-type bound is valid in all spatial dimensions, and for all $\mu\in\bbR$. The latter is one of the advantages of working immediately in infinite volume, since otherwise the size of the domain $\Lambda_L$ yields a lower bound on the energy levels $\mu$ that can be considered.

As already seen in other applications of the DFM landscape function, this bound where the volume is not coarse-grained reflects the fact that the transformation
\begin{equation*}
-\Delta + V\mapsto-\frac{1}{u^2}\mathrm{div} \, u^2\nabla +\frac{1}{u}
\end{equation*}
`transfers' some of the kinetic energy to the effective potential.

One may wonder about the relationship of the effective potential with the semi-classical limit. Not surprisingly, one partial answer is provided by microlocal analysis. Indeed, the effective potential is given by the resolvent acting on the constant function. The inverse of $-\Delta + V$ is a pseudo-differential operator whose symbol is given to highest order by $\frac{1}{\vert \xi\vert ^2 + V(x)}$. In this approximation, we conclude that, formally,
\begin{equation*}
u(x) \simeq \int_{\bbR^d}\frac{\mathrm{e}^{\mathrm{i}\xi\cdot x}}{\vert \xi\vert ^2 + V(x)}\delta(\xi) d\xi = \frac{1}{V(x)}.
\end{equation*}
In other words, the effective potential $\frac{1}{u}$ is equal to the physical potential $V$ up to lower order corrections in the sense of microlocal analysis. It is precisely these corrections however that remove the need for coarse-graining. 

We conclude this introduction by commenting on the specific case $V(x,y) = x^2 y^2$ in two dimensions. In~\cite{Si1}, Simon provided five proofs that this operator has discrete spectrum with strictly positive first eigenvalue, that the number of eigenvalues below any fixed energy $\mu$ follows the Fefferman-Phong estimate (in particular, the coarse-grained volume is finite), and established the precise asymptotics in \cite{Si2}. These results are particularly remarkable given that the Lebesgue measure of $ \{ (x,y) \in \mathbb{R}^2 \ : \ x^2 y^2 \leq \mu  \}$ is infinite for every $\mu>0$. This shows in particular that the measure of the sublevel sets of the potential do not capture the spectral information, unlike the coarsed-grained volume associated with the potential $x^2y^2$. Using pseudo-differential calculus \cite{Rob} obtains similar results for more general degenerate polynomials. Our result removes the need for coarse graining, provided $V$ is replaced with the DFM effective potential $\frac{1}{u}$; Since it is simple to see that the effective potential is confining in this case (by Corollary \ref{cor:Polynomials}), we obtain yet another proof of discreteness of the spectrum of $-\Delta + x^2 y^2$ in $\bbR^2$.

%%%%%%%%%%%%%%%%%%%%%%%%%%%%%%%%%%%%%%%%%%%%%%%%%%%%%%%%%%%%%%%		
\section{Main results}

We denote by $u$ a particular weak solution of
\begin{equation*}
(-\Delta + V)u = 1
\end{equation*}
in $\bbR^d$, which can be realized as the pointwise limit of specific Lax-Milgram solutions. It will be constructed in details in Section \ref{sect:infinite volume}. Our main result is that the volume of
\begin{equation*}
\left\{x\in\bbR^d:\frac{1}{u(x)}\leq\mu\right\}
\end{equation*}
is comparable to the rank of the spectral projection of (the Friedrich's extension of) $-\Delta +V$ at energies less than or equal to $\mu$.  
		
		\begin{theorem} \label{thm:main}
			Assume that $V\geq 0,V\not\equiv 0$ satisfies the following conditions:
				\begin{enumerate}
				\item  (Kato-type condition) There exists $\cK, \delta>0$, such that
				\begin{equation} \label{assump:Kato}
					\frac{1}{r^{d-2+\delta}}\int_{B(x,r)} V(y) dy \leq \cK \frac{1}{R^{d-2+\delta}} \int_{B(x,R)} V(y) dy
				\end{equation}
 for all $x\in \mathbb{R}^d$ and all $r,R$ with $0<r<R$.
 				\item (Doubling condition) There exists $\cD>0$ such that
				\begin{equation} \label{assump:doubling}
					\int_{B(x,2r)} V(y)dy \leq \cD \left( \int_{B(x,r)} V(y) dy + r^{d-2} \right)
				\end{equation}
				for all $x\in \mathbb{R}^d$ and all $r>0$.
			\end{enumerate}
			We denote by $H$ the Friedrichs extension of the positive symmetric operator $$-\Delta + V$$ defined on $C_c^\infty(\mathbb{R}^d)$. Let $\mathcal{N}^V(\mu)$ be the rank of the spectral projection $\mathbbm{1}_{(\infty,\mu]}(H)$. Then there exist constants $c, C>0$ such that for all $\mu \in \mathbb{R}$,
			\begin{equation} \label{eq:MainBathtub}
				(c\mu)^{\frac{d}{2}} \caV(c\mu)\leq \caN^V(\mu)\leq (C\mu)^{\frac{d}{2}}\caV(C\mu),
			\end{equation}
			where 
			\begin{equation*}
				\caV(\mu) = \int_{\left\{x\in\bbR^d:\frac{1}{u(x)}\leq\mu\right\}}dx.
			\end{equation*}
			The constants $c, C$ depend only on $\cK,\cD,\delta$ and the spatial dimension $d$.
		\end{theorem}
	
		The assumption \eqref{assump:Kato} is a scale-invariant variant of the standard Kato condition. For $d\geq 3$ one obtains via Fubini's theorem that the condition \eqref{assump:Kato} is equivalent to
		\begin{align*}
			\int_{B(x,R)} \frac{V(y)}{\vert x-y\vert^{d-2}} dy \leq \frac{C}{R^{d-2}} \int_{B(x,R)} V(y) dy 
		\end{align*}
		for all $0<R$, all $x\in \mathbb{R}^d$ and some $C$ independent of $x,R$. Conditions \eqref{assump:Kato},\eqref{assump:doubling} are satisfied by potentials in the reverse Hölder class $(RH)_{d/2}$ (see \cite{She1}). In particular, this include nonnegative polynomials and fractional power functions $\vert x \vert^\alpha$ for $\alpha>-2$ for $d\geq 3$. On the other hand, potentials with compact support or exponential growth violate \eqref{assump:Kato}, respectively \eqref{assump:doubling}.

		Combining the above theorem with the property that $u$ varies slowly, we further derive an analogous result to \cite[Corollary 0.11]{She2}.
		
		\begin{cor} \label{cor:confining}
			Let $V$ be as in Theorem \ref{thm:main}. Then the spectrum of $H$ is discrete if and only if $\lim_{R \rightarrow \infty} \Vert u \Vert_{L^\infty(\mathbb{R}^d\setminus B(0,R))} =0$.
		\end{cor}
		We will present the proofs in Section \ref{sect:bathtub}, while in Section \ref{sect:cor} we concentrate on the case where $V$ is a polynomial. For polynomial potentials one can further analyze the landscape function. In particular, one has the following:
		
		\begin{cor} \label{cor:Polynomials}
			Let $V$ be a polynomial that is bounded from below. Then the spectrum of $H$ is discrete if and only if none of the directional derivatives of $V$ vanishes identically.
		\end{cor}
		
		If the condition of the last corollary is violated, then the corresponding operator has no eigenvalues. Indeed, after conjugating with a suitable rotation, we can assume without loss of generality that the polynomial does not depend on the last variable. We define $\widetilde{V}(x_1, \dots, x_{d-1}) = V(x_1, \dots, x_d)$. By taking a Fourier transform in the last variable, we get that the Friedrichs extension $H$ of $-\Delta +V$ is unitarily equivalent to the direct integral
		\begin{align*}
			\int_{\mathbb{R}}^{\oplus} (p^2+ \widetilde{H}) dp
		\end{align*}
		where $\widetilde{H}=-\Delta_{\mathbb{R}^{d-1}} + \widetilde{V}$. It follows from \cite[Theorem XIII.85]{RS4} that $\sigma(H)=[\min \sigma(\widetilde{H}), \infty)$ and $H$ admits no eigenvalue.
		
		Finally, we point out that Theorem \ref{thm:main} together with \eqref{eq:equivum} below recover the result of Shen \cite[Theorem 0.9]{She2} in our class of potentials and in the absence of a magnetic potential.

		\section{Existence of the landscape function in infinite volume} \label{sect:infinite volume}
		In this section we show the existence of the landscape function in infinite volume and establish some estimates of the landscape function in terms of the Fefferman-Phong-Shen maximal function. In \cite[Theorem 1.18, Theorem 1.31]{P} Poggi proves this for $d\geq 3$. We briefly recall the construction and explain how to extend this to the case $d=1,2$. For this we will rely on extensions of results for $d\geq 3$ due to Shen (see \cite[Proposition 1.8]{She4}). One of the key objects is the Fefferman-Phong-Shen maximal function $m(\cdot, V)$, which is defined as
		  \begin{equation} \label{eq:defm(x,V)}
			\frac{1}{m(x,V)} = \sup \left\{ r >0 \ : \ \frac{1}{r^{d-2}} \int_{B(x,r)} V(y) dy\leq \cD \right\},
		\end{equation}
		where $\cD$ is the constant in \eqref{assump:doubling}. This maximal function satisfies the following properties.
		
		\begin{lemma} \label{lm:m}
			Let $V$ satisfy the conditions of Theorem \ref{thm:main}. Then we have
			\begin{enumerate}
				\item $0<m(x,V) < \infty$ for every $x\in \mathbb{R}^d.$
				\item  For every $C'$, there exists $C$, depending only on $\cK,\cD, \delta$ and $C'$, such that
				\begin{equation} \label{eq:Harnackm}
					C^{-1} m(x,V) \leq m(y,V) \leq C m(x,V)
				\end{equation}
				for all $x,y\in \mathbb{R}$ with $\vert x - y \vert \leq \frac{C'}{m(x,V)}$.
				\item There exists $k_0, C>0$, depending only on $C_K,C_D, \delta$ and $d$, such that for all $x,y\in \mathbb{R}^d$ we have
				\begin{equation} \label{eq:upperboundm}
					m(x,V) \leq C m(y,V) \left( 1 + \vert x -y\vert m(y,V) \right)^{k_0}.
				\end{equation}
				\item  Let $d\leq 2$ and let  $\widetilde{V}(x,t) = V(x)$ for all $(x,t)\in \mathbb{R}^d \times \mathbb{R}$. Then for all $(x,t)\in \mathbb{R}^{d+1}$ and all $0<r<R$,
				\begin{align*}
					\frac{1}{r^{d-1+\delta}}\int_{B((x,t),r)} \widetilde{V}(z)dz \leq \cK \sqrt{2}^{d-1+\delta} \frac{1}{R^{d-1+\delta}} \int_{B((x,t),R)} \widetilde{V}(z) dz.
				\end{align*}
				Furthermore, 
				\begin{align*}
					\int_{B(x,2r)} \widetilde{V}(z) dz \leq 4\cD \left(\int_{B((x,t),r)} \widetilde{V}(z) dz + r^{d-1}\right)
				\end{align*}
				for all $(x,t)\in \mathbb{R}^{d+1}$ and all $r>0$.
			
				Finally, there exists $C>0$ depending on $\cK, \cD, \delta$ and $d$ such that 
				\begin{equation} \label{eq:dimension reduction}
					C^{-1}m(x,V) \leq m((x,t), \widetilde{V}) \leq C m(x,V)
				\end{equation}
				for all $(x,t)\in \mathbb{R}^d \times \mathbb{R}$.
			\end{enumerate}
		\end{lemma}
\noindent Note that in~(\ref{eq:dimension reduction}), the exponent in the maximal function involving $V$ is $d$, while it is $d+1$ in the one involving $\tilde V$.
		\begin{proof}
			First we note that \eqref{assump:Kato} yields for all $0<r<R$
			\begin{align*}
				R^\delta\frac{1}{r^{d-2}} \int_{B(x,r)} V(y) dy \leq \cK r^\delta \frac{1}{R^{d-2}} \int_{B(x,R)} V(y)dy.
			\end{align*}
			Thus, $\lim_{r\rightarrow 0^+} r^{2-d} \int_{B(x,r)} V(y)dy=0$ and $\lim_{R\rightarrow \infty} R^{2-d} \int_{B(x,R)} V(y) dy = \infty$.
			This implies that $0<m(x,V)<\infty$. 
			
The validity of $2.$ and $3.$ for $d\geq 3$ is proved in \cite[Proposition 1.8]{She4}. Hence, it suffices to prove $4.$ for $2.$ and $3.$ to hold for all $d\geq 1$.
			
The fact that $\widetilde{V}$ satisfies both the Kato-type and the doubling condition are simple computations. Since $V$ satisfies the doubling condition in dimension $d$, 
			\begin{align*}
				\int_{B((x,t),2r)} \widetilde{V}(z) dz
				&\leq \int_{-2r}^{2r} \cD \left( \int_{B(x,\sqrt{r^2-s^2/4})} V(y) dy + \left(r^2-s^2/4\right)^{(d-2)/2} \right) ds \\
				&= 2\cD \left( \int_{B((x,t),r)} \widetilde{V}(z) dz + r^{d-1} \int_{-1}^1 \left(1-\sigma^2\right)^{(d-2)/2} d\sigma \right),
			\end{align*}
which yields the claim upon noting that the last integral is bounded above by $\pi$. 
			For the Kato-type condition, we first consider the case $0<\sqrt{2}r<R$. Then we have
			
			\begin{align*}
				\int_{B((x,t),r)} \widetilde{V}(z) dz
				&\leq \int_{B(x,r) \times (t-r,t+r)} \widetilde{V}(z)dz
				= 2r\int_{B(x,r)} V(y) dy \\
				&\leq \cK (2r) \left(\frac{r}{R/\sqrt{2}}\right)^{d-2+\delta} \int_{B(x,R/\sqrt{2})} V(y) dy \\
				&= \cK \left( \sqrt{2}\frac{r}{R} \right)^{d-1+\delta} \int_{B(x,R/\sqrt{2}) \times (t-R/\sqrt{2}, t+R/\sqrt{2})} \widetilde{V}(z) dz \\
				&\leq \cK  \left( \sqrt{2}\frac{r}{R} \right)^{d-1+\delta} \int_{B((x,t),R)} \widetilde{V}(z) dz.
			\end{align*}
			The bound is immediate if, on the other hand $\frac{R}{\sqrt 2}\leq r < R$, since
			\begin{align*}
				\int_{B((x,t),r)} \widetilde{V}(z) dz \leq \int_{B((x,t),R)} \widetilde{V}(z) dz
				\leq \left(\sqrt{2}\frac{r}{R}\right)^{d-1+\delta} \int_{B((x,t),R)} \widetilde{V}(z) dz.
			\end{align*}
			
			To show \eqref{eq:dimension reduction} we introduce the following maximal function
			\begin{align*}
					\frac{1}{m_Q(x,V)} = \sup \left\{ r >0 \ : \ \frac{1}{r^{d-2}} \int_{Q(x,r)} V(y) dy\leq \cD \right\}
			\end{align*}
			which is defined over cubes $Q(x,r)$ centered at $x$ and of sidelength $r$, rather than over balls. Clearly \eqref{eq:dimension reduction} holds true with $C=1$ for $m$ replaced by $m_Q$. Thus, we only need to show that $m$ and $m_Q$ are equivalent. For $d\leq 2$, the inclusion $Q(x,r)\subset B(x,r)$ and the positivity of $V$ yield immediately $\frac{1}{m(x,V)} \leq \frac{1}{m_Q(x,V)}$, and hence $m_Q(x,V)\leq m(x,V)$. Reciprocally, let $x\in\bbR^d$ and let $r=\frac{2}{m(x,V)}$. Then for any $R>\cK^{1/\delta}r$, 
\begin{align*}
\cD &= \frac{1}{(r/2)^{d-2}}\int_{B(x,r/2)}V(y)dy
\leq \cK \left(\frac{r}{R}\right)^\delta\frac{1}{(R/2)^{d-2}}\int_{B(x,R/2)}V(y)dy \\
&\leq  \frac{1}{R^{d-2}}\int_{Q(x,R)}V(y)dy
\end{align*}
where we used~(\ref{assump:Kato}) in the second inequality, and the fact that $B(x,R/2)\subseteq Q(x,R)$ in the third. It follows that $\frac{1}{m_Q(x,V)}\leq \cK^{1/\delta}r$ and so $m_Q(x,V)\geq \frac{m(x,V)}{2\cK^{1/\delta}}$.
		\end{proof}
		
Let $f\in L^\infty(\mathbb{R}^d)$ be compactly supported. We call $u_f$ a \emph{Lax-Milgram solution} of
\begin{equation*}
(-\Delta + V)u = f
\end{equation*}
if $u_f$ is in the form domain $\mathcal{H}$ of $H$ and 
		\begin{align*}
			\int_{\mathbb{R}^d} \left( \nabla u_f(y) \cdot\nabla v(y) + V(y) u_f(y)v(y)\right) dy = \int_{\mathbb{R}^d} f(y) v(y) dy
		\end{align*} 
		for all $v\in\mathcal{H}$. Note that
		\begin{align*}
			\mathcal{H} = \left\{ v \in H^1(\mathbb{R}^d) \ : \ \int_{\mathbb{R}^d} V(y) \vert v(y)\vert^2 dy < \infty \right\},
		\end{align*}
see \cite[Theorem 8.2.1]{Dav}.
		The following proposition yields estimates for Lax-Milgram solutions.
		
		\begin{prop}{\cite[Theorem 0.8, Theorem 2.16]{She4}}\label{Prop:LM solution}
			Let $d\geq 3$ and assume $V$ satisfies the conditions of Theorem \ref{thm:main}. For every $x\in \mathbb{R}^d$ there exists a function $\Gamma_V(x,\cdot) \in L^p_{\mathrm{loc}}(\bbR^d)$, for $1<p<d/(d-2)$, such that for all $f\in L^\infty(\mathbb{R}^d)$ with compact support and $f\geq 0$, the unique Lax-Milgram solution $u_f$ of $(-\Delta +V)u = f$ can be written as
			\begin{equation} \label{eq:integral rep}
				u_f(x) = \int_{\mathbb{R}^d} \Gamma_V(x,y) f(y) dy
			\end{equation}
			for almost every $x\in \mathbb{R}^d$.
			Furthermore, one has the kernel estimate
			\begin{equation} \label{eq:KernelEstimate}
				\frac{ce^{-\varepsilon (1+\vert x-y\vert m(x,V))^{k_0+1}}}{\vert x -y\vert^{d-2}} \leq \Gamma_V(x,y) \leq \frac{Ce^{-\varepsilon (1+\vert x-y\vert m(x,V))^{1/(k_0+1)}}}{\vert x -y\vert^{d-2}}.
			\end{equation}
		\end{prop}
		\begin{proof}
			Existence and uniqueness of Lax-Milgram solution follows directly from the Lax-Milgram theorem on the form domain $\mathcal{H}$ equipped with its standard inner product $\langle v, w \rangle_\mathcal{H}= \langle \nabla v, \nabla w \rangle_{L^2(\mathbb{R}^d)}+ \langle \sqrt{V}v, \sqrt{V}w \rangle_{L^2(\mathbb{R}^d)}$. The representation of the Lax-Milgram solution in terms of the integral kernel $\Gamma_V$ was shown in \cite[Theorem 2.16]{She4} and the estimate in terms of the Fefferman-Phong-Shen maximal function follow from \cite[Theorem 3.11, Remark 3.21, Theorem 4.15]{She4}. 
		\end{proof}

In what follows, we will also consider \emph{weak solutions} of $(-\Delta+V)u=f$ for $f\in L_\mathrm{loc}^1(\mathbb{R}^d)$, namely a function $u_f$ such that
		\begin{align*}
			\int_{\mathbb{R}^d} u_f(y) \left(-\Delta \varphi(y) + V(y) \varphi(y)\right) dy = \int_{\mathbb{R}^d} f(y) \varphi(y) dy
		\end{align*} 
		for all $\varphi\in C_c^\infty(\mathbb{R}^d)$.
		
		We shall now construct the landscape function in infinite volume. This is an alternative and simpler approach, valid in the present setting, than that of~\cite[Theorem 1.18]{P}.
		
		\begin{prop}\label{Prop:Equivalence}
			Let $V$ be as in Theorem \ref{thm:main}. Then there exists constants $c,C>0$ depending only on $\cK, \cD,\delta$ and $d$, and a weak solution $u\in H_\mathrm{loc}^1(\mathbb{R}^d)\cap C^0(\mathbb{R}^d)$ of $(-\Delta +V) u =1$ such that
			\begin{equation} \label{eq:equivum}
				\frac{c}{m(x,V)^2} \leq u(x) \leq \frac{C}{m(x,V)^2},
			\end{equation}
			 for almost every $x\in \mathbb{R}^d$.
 \end{prop}
For later purposes, we immediately note that the proof of the proposition yields the following `finite volume' result. If $d\geq 3$ we denote by $u_L$ the Lax-Milgram solutions of 
			 \begin{equation} \label{eq:uL}
			 	(-\Delta +V) u_L = \mathbbm{1}_{B(0,L)}.
			 \end{equation}
Then
			 \begin{equation} \label{eq:equivuLm}
			 	\frac{c}{m(x,V)^2} \leq u_L(x) \leq \frac{C}{m(x,V)^2},
			 \end{equation}
		 		for almost every $x\in \mathbb{R}^d$. We remark that all the results work equally well if we replace the indicator function over balls by indicator function over other compact sets $\{\Omega_L:L\in\bbN\}$ such that $\Omega_L \subseteq \Omega_{\widetilde{L}}$ for $L\leq \widetilde{L}$ and $\bigcup_{L\geq 1} \Omega_L=\mathbb{R}^d$.

		\begin{proof}
			First we consider the case $d\geq 3$. Denote by $u_L$ the Lax-Milgram solution \eqref{eq:uL}
			given by Proposition~\ref{Prop:LM solution}. As $\mathbbm{1}_{B(0,L_2)}-\mathbbm{1}_{B(0,L_1)}\geq 0$ for $L_2\geq L_1$, we get from \eqref{eq:integral rep} that $(u_L)_{L\geq 1}$ is monotone increasing almost everywhere. On the other hand, it is essentially bounded since
			\begin{align*}
				0\leq u_L(x) \leq \frac{C}{m(x,V)^2} \int_{\mathbb{R}^d} \frac{e^{-\varepsilon(1+\vert y \vert)^{1/(k_0+1)}}}{\vert y \vert^{d-2}} dy = \frac{\widetilde{C}}{m(x,V)^2}
			\end{align*}
for almost every $x\in \mathbb{R}^d$ by~(\ref{eq:integral rep},\ref{eq:KernelEstimate}). Thus, we can define
			\begin{align*}
				u(x) = \lim_{L\rightarrow \infty} u_L(x).
			\end{align*}
			As $u_L$ are Lax-Milgram solutions of $(-\Delta + V)u_L = \mathbbm{1}_{B(0,L)}$, one easily checks that $u$ is a weak solution of $(-\Delta +V) u= 1$. The lower bound for $u$ follows from the lower bound in \eqref{eq:KernelEstimate}.
			
			We show now that $u\in H_\mathrm{loc}^1(\mathbb{R}^d)$ for $d\geq 3$. Fix any ball $B\subseteq \mathbb{R}^d$ and a smooth cut-off function $\chi_B \in C_c^\infty(\mathbb{R}^d)$ such that $\chi_B\equiv 1$ on $B$. As $\chi_B u_L\in \mathrm{dom}(H^{1/2})$ and $u_L$ is a Lax-Milgram solution of \eqref{eq:uL}, we get
			\begin{align*}
				\int_{\mathbb{R}^d} \nabla u_L \cdot \nabla \left( \chi_B u_L\right) + \int_{\mathbb{R}^d} V u_L (\chi_B u_L) = \int_{\mathbb{R}^d} u_L \chi_B \mathbbm{1}_{B(0,L)}.
			\end{align*}
The product rule for Sobolev functions yields
			\begin{align*}
				\int_{\mathbb{R}^d} \nabla u_L \cdot \nabla (\chi_B u_L) = \int_{\mathbb{R}^d} \vert \nabla u_L \vert^2 \chi_B + \int_{\mathbb{R}^d} \nabla u_L \cdot \left( \nabla \chi_B \right) u_L.
			\end{align*}
			Using integration by parts for the second term on the RHS yields
			\begin{align*}
				\int_{\mathbb{R}^d} \nabla u_L \cdot \left( \nabla \chi_B \right) u_L = - \int_{\mathbb{R}^d} \left( (\Delta \chi_B)u_L^2 + u_L (\nabla\chi_B) \cdot \nabla u_L \right).
			\end{align*}
			Thus, we get
			\begin{align*}
				\int_B \vert \nabla u_L \vert^2 \leq \int_{\mathbb{R}^d} \chi_B\vert \nabla u_L \vert^2 = \int_{\mathbb{R}^d} u_L \chi_B \mathbbm{1}_{B(0,L)} - \int_{\mathbb{R}^d} V \chi_B u_L^2 + \frac{1}{2} \int_{\mathbb{R}^d} (\Delta \chi_B) u_L^2.
			\end{align*}
			Hence, there exists a constant $C>0$ depending only on the dimension such that
			\begin{align*} %\label{eq:Caccioppoli}
				\int_B (\vert\nabla u_L \vert^2 + V u_L^2) \leq \int_{2B} u_L + C\int_{2B} u_L^2 
			\end{align*}
 for all balls $B$ and all $L>0$.
			By \eqref{eq:equivuLm} and \eqref{eq:upperboundm} we get that $(\nabla u_L)_{L\geq 1}$ is uniformly bounded in $L^2(B)$ for fixed $B$. Therefore, by Banach-Alaoglu, there exists a subsequence $u_{L_k}$ converging weakly to some $g_B\in L^2(\mathbb{R}^d)$. One readily checks that $g_B$ is the weak derivative of $u$ and hence $u\in H_\mathrm{loc}^1(\mathbb{R}^d)$.
			
			Next we consider the case $d=2$. For this we use Hadamard's method of descent. Recall that $\widetilde{V}(x,t)=V(x)$ for $(x,t) \in \mathbb{R}^2 \times \mathbb{R}$. By Lemma \ref{lm:m} the function $\widetilde{V}$ satisfies (\ref{assump:Kato},\ref{assump:doubling}), and therefore the first part yields a weak solution $\widetilde{u}$ of $(-\Delta + \widetilde{V}) \widetilde{u} =1$ on $\mathbb{R}^3$. Let $\alpha \in \mathbb{R}$. One readily checks that $\widetilde{v_L}(x,t)=\widetilde{u_L}(x,t+\alpha)$ is a Lax-Milgram solution of $(-\Delta +\widetilde{V}) \widetilde{v_L} = \mathbbm{1}_{B((0,0,\alpha),L)}$. Thus, by \eqref{eq:integral rep}, we have
			\begin{align*}
				\widetilde{u_L}(x,t+\alpha) = \int_{\mathbb{R}^3} \Gamma_{\widetilde{V}}((x,t), y) \mathbbm{1}_{B((0,0,\alpha),L)}(y) dy.
			\end{align*}
			As $\Gamma_{\widetilde{V}}((x,t), \cdot) \in L^1(\mathbb{R}^3)$ by \eqref{eq:KernelEstimate}, we get by dominated convergence			
			\begin{align*}
				\widetilde{u}(x,t+\alpha) = \lim_{L\rightarrow \infty} \widetilde{u_L}(x, t+\alpha) = \int_{\mathbb{R}^3} \Gamma_{\widetilde{V}}((x,t), y) dy
				= \widetilde{u}(x,t).  
			\end{align*}
			Hence, for almost every $x\in \mathbb{R}^2$ there exists $C_x$ such that for almost every $t\in \mathbb{R}$ we have $\widetilde{u}(x,t)=C_x$ and we define $u$ on $\bbR^2$ by $u(x) = C_x$. Let $\varphi\in C_c^\infty(\mathbb{R}^2)$ and $\psi\in C_c^\infty(\mathbb{R})$ with $\int_\mathbb{R} \psi(t) dt=1$. Then, as $\int_\mathbb{R} \psi''(t) dt =0$ and $\widetilde{u}$ is a weak solution of $(-\Delta + \widetilde{V}) \widetilde{u}=1$ on $\mathbb{R}^3$, we get
			\begin{align*}
				\int_{\mathbb{R}^2} u(x) (-\Delta +V(x))\varphi(x) dx
				&= \int_{\mathbb{R}^3} \widetilde{u}(x,t) (-\Delta +\widetilde{V}(x,t)) (\varphi(x) \psi(t)) dx dt \\
				&= \int_{\mathbb{R}^3} \varphi(x) \psi(t) dt = \int_{\mathbb{R}^2} \varphi(x)dx.
			\end{align*}
			Therefore, $(-\Delta +V) u=1$ is a weak solution on $\mathbb{R}^2$. The inequality \eqref{eq:equivum} follows from \eqref{eq:dimension reduction}. As shown before, we have $\widetilde{u}\in H_\mathrm{loc}^1(\mathbb{R}^3)$ and $\vert \nabla u(x)\vert =\vert \nabla \widetilde{u}(x,t)\vert$ as $u(x,t)$ is independent of $t$. Hence $u\in H_\mathrm{loc}^1(\mathbb{R}^2)$.
					
			The case $d=1$ follows similarly as the case $d=2$.
			
			Finally, continuity follows from \cite[Corollary 1.5]{LS}.
		\end{proof}
	
We point out that the weak solution $u$ constructed above does  in general not belong to the form domain of $H$, and we will therefore often have to work with the Lax-Milgram solution $u_L$ instead of $u$. If $V$ is a polynomial, the maximal function $m(\cdot,V)$ is equivalent to the function introduced in~\cite{Sm,Zh} 
 \begin{equation} \label{SmithZhongMaxFunction}
     M(x,V) = \sum_{\alpha \in \mathbb{N}_0^n} \vert \partial^\alpha V(x) \vert^{1/(\vert \alpha \vert +2)},
 \end{equation}
see~(\ref{eq:Mequivm}) below. The sum is of course finite for a polynomial. We now consider $V(x) = \vert x \vert^2$ on $\mathbb{R}^d$. Then $M(x,V)$ is comparable to $1+\vert x \vert$
and hence, by \eqref{eq:equivum} and \eqref{eq:Mequivm}, $u(x)$ is comparable to $(1+\vert x\vert)^{-2}$ which is not square integrable for $d > 2$. In Lemma \ref{lm:formdomain} we show that the landscape function, for polynomial potentials, belongs to the form domain if and only if the landscape function is integrable.

		The equivalence of the landscape function and the Fefferman-Phong-Shen maximal function exhibited in Proposition~\ref{Prop:Equivalence} allows one to prove a Harnack inequality for the landscape function, see also \cite[Corollary 1.38]{P} for the case $d\geq 3$. 
		
		\begin{cor}
			Let $V$ be as in Theorem \ref{thm:main}. Then there exists a constant $\cH\geq 1$, depending only on $\cK, \cD,\delta$ and $d$, such that for almost every $x\in \mathbb{R}^d$ and almost every $y\in Q(x,2\sqrt{u(x)})$ we have
			
			\begin{equation} \label{eq:Harnack}
				\cH^{-1} u(x) \leq u(y) \leq \cH u(x).
			\end{equation}
		\end{cor}
		\begin{proof}
			This follows immediately from \eqref{eq:Harnackm} and \ref{eq:equivum}.
		\end{proof}

		\section{Proof of Theorem \ref{thm:main}} \label{sect:bathtub}

In this section we show that we can estimate the rank $\mathcal{N}^V(\mu)$ of the spectral projection of $H$ in terms of the measure of the sublevel set $\mathcal{V}(\mu)$ of the effective potential~$\frac{1}{u}$, both defined in Theorem \ref{thm:main}.

For this we introduce two types of coarse-grained volumes. A box of sidelength $\ell$ is a set of the form $\times_{i=1}^d[a_i,b_i]$ where $b_i-a_i = \ell$. For any $\ell>0$, we consider a collection $\caQ_\ell$ of boxes of sidelength $\ell$ such that $\bigcup_{Q\in\caQ_\ell}Q = \bbR^d$ and $\mathring Q \cap \mathring Q' = \emptyset$ whenever $Q\neq Q'$. We define for any $\mu>0$
\begin{align*}
	N(\mu) = \left\vert \left\{ Q\in \caQ_{\mu^{-1/2}} \ : \ \esi \, \frac{1}{u} \leq \mu \right\}\right\vert
\end{align*}
and
\begin{align*}
	n(\mu) = \left\vert \left\{ Q\in \caQ_{\mu^{-1/2}} \ : \ \ess \, \frac{1}{u} \leq \mu \right\}\right\vert,
\end{align*}
where $\operatorname{ess} \inf, \operatorname{ess} \sup$ denote the essential infimum, respectively the essential supremum.

For the class of potentials considered here, namely those satisfying the Kato-type and doubling conditions, both coarse-grained volumes are directly related to the measure $\caV(\mu)$ of the sublevel set.

\begin{lemma} 
	Let $V$ satisfy the conditions of Theorem~\ref{thm:main}. Then 
	\begin{align*} %\label{eq:ComparisonnNS}
		n(\mu) \leq \mu^{d/2} \caV(\mu) \leq N(\mu) \leq n(\cH \mu)
	\end{align*}
	for all $\mu \in \mathbb{R}$.
\end{lemma}
\begin{proof}
	The first two inequalities are immediate as, up to null sets, $n(\mu)/\mu^{d/2}$ is the measure of all boxes that are strictly contained in the sublevel set $\{1/u \leq \mu \}$ and $N(\mu)/\mu^{d/2}$ is the measure of all the boxes that intersect the sublevel set.
	
	Let $Q$ be a box such that $\inf_Q \, 1/u \leq \mu$, then by \eqref{eq:Harnack} we have $\sup_Q \, 1/u \leq \cH \mu$. Hence, $N(\mu) \leq n(\cH \mu)$. 
\end{proof}

We now turn to the proof of the main theorem, namely the bounds \eqref{eq:MainBathtub}. Our arguments are variational and adapted from the proofs of~\cite{DFM}, which are themselves inspired by Fefferman-Phong \cite{FP}. We start with the upper bound.

\begin{lemma}
	\label{lma:UpperBound}
	Let $V$ satisfy the conditions of Theorem~\ref{thm:main}. Then
	\begin{align*}
		\mathcal{N}^V(\mu) \leq N(C\mu)
	\end{align*}
	for all $C>\max\{ 2, \frac{2d}{\pi^2}\}$ and all $\mu\in \mathbb{R}$.
\end{lemma}
\begin{proof}
	In order to have that $\mathcal{N}^V(\mu) \leq N$ it suffices, by the Min-Max Principle (see \cite[Theorem XIII.2]{RS4}), to exhibit a subspace $\mathcal{H}_N\subseteq \mathrm{dom}(H^{1/2})$ with codimension at most $N$ such that
	\begin{align*}
		\int_{\mathbb{R}^d} \left( \vert \nabla v \vert^2 + V \vert v \vert^2 \right) > \mu \int_{\mathbb{R}^d} \vert v \vert^2
	\end{align*}
	for all $v \in \mathcal{H}_N$. Let $\mathcal{F}$ be the collection of boxes such that 
	\begin{align*}
		\mathcal{F} = \left\{ Q\in \caQ_{(C\mu)^{-1/2}} \ : \ \esi \, \frac{1}{u} \leq C\mu \right\},
	\end{align*}
	where $C>0$ will be chosen later, and let
	\begin{align*}
		\mathcal{H}_N = \left\{ v \in \mathrm{dom}(H^{1/2}) \ : \  \int_Q v =0 \quad \forall Q\in \mathcal{F} \right\}.
	\end{align*}
	Since the cubes are disjoint, the codimension of $\mathcal{H}_N$ is equal to $\vert \mathcal{F}\vert =N(C\mu)$. 
	
	First we want to show that
	
	\begin{equation} \label{eq:lowerbnd}
		\langle (-\Delta +V) \varphi, \varphi \rangle_{L^2(\mathbb{R}^d)} \geq \langle \frac{1}{u} \varphi, \varphi \rangle_{L^2(\mathbb{R}^d)}
	\end{equation}
	for all $\varphi\in C_c^\infty(\mathbb{R}^d)$. We start by considering $d\geq 3$. Denote by $u_L$ the Lax-Milgram solution of \eqref{eq:uL}. By \eqref{eq:upperboundm}, \eqref{eq:equivuLm} we know that $1/u_L \in L_\mathrm{loc}^\infty(\mathbb{R}^d) \cap H_\mathrm{loc}^1(\mathbb{R}^d)$, using the chain rule for Sobolev functions \cite[Theorem 6.16]{LL}. This readily implies that $\vert \varphi\vert^2/u_L$ is in the form domain of $H$ for all $\varphi\in C_c^\infty(\mathbb{R}^d)$. As $u_L$ is a Lax-Milgram solution of \eqref{eq:uL}, we get
	\begin{align*}
		\int_{\mathbb{R}^d} \left( \nabla u_L \cdot \nabla \left(\frac{\vert \varphi\vert^2}{u_L}\right) + V u_L  \frac{\vert\varphi\vert^2}{u_L} \right) = \int_{\mathbb{R}^d} \frac{\mathbbm{1}_{B(0,L)}}{u_L} \vert \varphi \vert^2.
	\end{align*}
	Furthermore, using the product rule \cite[Lemma 7.4]{LL} yields
	\begin{align*}
		\nabla u_L \cdot \nabla (\vert \varphi \vert^2/u_L) = \vert \nabla \varphi \vert^2 - u_L^2 \vert \nabla (\varphi/u_L)\vert^2.
	\end{align*}
	Combining the last two equalities and taking $L\rightarrow \infty$ implies \eqref{eq:lowerbnd} for $d\geq 3$.
	
	For $d\leq 2$ we set $\widetilde{V}(x,t) = V(x)$ for all $(x,t)\in \mathbb{R}^d \times \mathbb{R}^{3-d}$ and denote by $\widetilde{u}$ the landscape function of $\widetilde{V}$. Let $\varphi\in C_c^\infty(\mathbb{R}^d)$, $\psi\in C_c^\infty(\mathbb{R}^{3-d})$ with $\int_{\mathbb{R}^{3-d}} \psi(t) dt =1$ and $(\varphi \otimes \psi)(x,t) = \varphi(x) \psi(t)$ for all $(x,t)\in \mathbb{R}^d \times \mathbb{R}^{3-d}$. Then we have by the previous computations for $d=3$
	\begin{align*}
		\langle (-\Delta + V) \varphi, \varphi \rangle_{L^2(\mathbb{R}^d)}
		&= \langle (-\Delta + \widetilde{V}) (\varphi \otimes \psi), \varphi \otimes \psi \rangle_{L^2(\mathbb{R}^3)}  \\
		&\geq \langle \frac{1}{\widetilde{u}} (\varphi \otimes \psi), \varphi \otimes \psi \rangle_{L^2(\mathbb{R}^3)}
		= \langle \frac{1}{u} \varphi, \varphi \rangle_{L^2(\mathbb{R}^d)}.
	\end{align*}

	The bound \eqref{eq:lowerbnd} extends, for all $d\geq 1$, by density of $C_c^\infty(\mathbb{R}^d )$ in the form domain of $H$ (see \cite[Theorem 8.2.1.]{Dav}) to all $v\in \mathrm{dom}(H^{1/2})$.
	This implies that
	\begin{align*}
		2\int_{\mathbb{R}^d} \left( \vert \nabla v\vert^2 + V \vert v \vert^2 \right)
		\geq \int_{\mathbb{R}^d} \left( \vert \nabla v\vert^2 + \frac{1}{u} \vert v \vert^2 \right)
	\end{align*}
	for all $v\in \mathrm{dom}(H^{1/2})$. With this, the statement of the lemma follows from the claim that if $v\in \mathcal{H}_N\setminus \{0\}$, then
	\begin{align*}
		\int_{\mathbb{R}^d} \left( \vert \nabla v\vert^2 + \frac{1}{u} \vert v \vert^2 \right) > 2\mu \int_{\mathbb{R}^d} \vert v \vert^2.
	\end{align*}
	We check this inequality using the partition into boxes. In any box $Q\notin \mathcal{F}$, we simply use the  bound $\min_Q 1/u > C\mu$. If $Q\in \mathcal{F}$, we recall that the integral of $v$ vanishes and use the Poincar\'e inequality with optimal constant $\frac{\pi^2}{d}(C\mu)$ since the boxes have sidelength $(C\mu)^{-1/2}$, see~\cite{PW}. Hence, the claimed lower bound holds for all $C>\max\{2,\frac{2d}{\pi^2}\}$ indeed.
	\end{proof}

Next we turn to the lower bound in \eqref{eq:MainBathtub}.

\begin{lemma}\label{lma:LowerBound}
	Let $V$ satisfy the conditions of Theorem~\ref{thm:main}. Then
	\begin{align*}
		n(\mu) \leq \mathcal{N}^V\left((1+(4\cH)^2)\mu\right)
	\end{align*}
	for all $\mu\in\bbR$. 
\end{lemma}
\begin{proof}

	For a lower bound $N\leq \mathcal{N}^V(C\mu)$ it suffices, again by the Min-Max Principle, to find a subspace $\mathcal{H}_N \subseteq \mathrm{dom}(H^{1/2})$ of dimension at least $N$ such that 
	\begin{align*}
	\int_{\mathbb{R}^d} \left( \vert \nabla v \vert^2 + V \vert v \vert^2 \right) \leq C\mu \int_{\mathbb{R}^d} \vert v \vert^2.
	\end{align*}
	We define
	\begin{align*}
		\mathcal{F} = \left\{ Q \in \mathcal{Q}_{\mu^{-1/2}} \ : \ \ess \, \frac{1}{u} \leq \mu \right\}.
	\end{align*}
	Furthermore, for a box $Q$ we pick $\chi_Q \in H^1(\mathbb{R}^d)$ with $0\leq \chi_Q \leq 1, \Vert \nabla \chi_Q \Vert_{L^\infty(\mathbb{R}^d)} \leq 4\mu^{1/2}$, $\chi_Q \equiv 1$ on $Q/2$ and $\chi_Q\equiv 0$ on $\mathbb{R}^d\setminus Q$ (a possible choice for $\chi_Q$ is to interpolate linearly from $\partial (Q/2)$ to $\partial Q$). Since the functions $\chi_Q u$ are non-zero and orthogonal to each other, the space
	\begin{align*}
		\mathcal{H}_N = \mathrm{span}\{ \chi_Q u \ : \ Q \in \mathcal{F} \}
	\end{align*}
	is of dimension $\vert \mathcal{F} \vert = n(\mu)$. 
	
	By Proposition \ref{Prop:Equivalence} we have $u\in H_\mathrm{loc}^1(\mathbb{R}^d)\cap L_\mathrm{loc}^\infty(\mathbb{R}^d)$ and thus $\chi_Q u$ is in the form domain of $H$. Using the product rule for Sobolev function \cite[Lemma 7.4]{LL} and the fact that $u$ solves the landscape equation  we get for all $\varphi, \psi \in C_c^\infty(\mathbb{R}^d)$ 
	\begin{align*}
		\langle \nabla(&\psi u), \nabla \varphi \rangle_{L^2(\mathbb{R}^d)}
		+ \langle \psi u, V \varphi \rangle_{L^2(\mathbb{R}^d)} \\
		&= \langle \nabla u, \nabla(\psi \varphi) \rangle_{L^2(\mathbb{R}^d)} + \langle u, V\psi \varphi \rangle_{L^2(\mathbb{R}^d)} - \langle \nabla u, (\nabla \psi) \varphi \rangle_{L^2(\mathbb{R}^d)} + \langle (\nabla \psi) u, \nabla \varphi \rangle_{L^2(\mathbb{R}^d)} \\
		&= \langle \psi, \varphi \rangle_{L^2(\mathbb{R}^d)}  - \langle \nabla u, (\nabla \psi) \varphi \rangle_{L^2(\mathbb{R}^d)} + \langle (\nabla \psi) u, \nabla \varphi \rangle_{L^2(\mathbb{R}^d)}.
	\end{align*}
	Now we pick a sequence $(\varphi_n)_{n\in \mathbb{N}} \subseteq C_c^\infty(\mathbb{R}^d)$ such that $\mathrm{supp}(\varphi_n) \subseteq 2Q$, $\sup_n \Vert \varphi_n \Vert_{L^\infty(\mathbb{R}^d)}<\infty$ and $\varphi_n \rightarrow \chi_Q u$ in $H^1(\mathbb{R}^d)$ and a similar approximation $\psi_n \rightarrow \chi_Q u$ and we get
	\begin{align*}
		\int_{\mathbb{R}^d} \left( \vert \nabla(\chi_Q u)\vert^2 + V \chi_Q^2 u^2 \right)
		= \int_{\mathbb{R}^d} \left( \chi_Q^2 u + \vert \nabla \chi_Q \vert^2 u^2 \right)
	\end{align*}
	and in turn
	\begin{align*} %\label{UpperBound}
		\int_{\mathbb{R}^d} \left( \vert \nabla (\chi_Q u) \vert^2 + V \chi_Q^2 u^2 \right)
		& \leq \left(\ess \, \frac{1}{u} \right) \int_{Q} \chi_Q^2 u^2 + 4^2 \mu  \int_{Q} u^2 \\
		& \leq \mu \left( \int_{Q} \chi_Q^2 u^2 + 4^2\int_{Q} u^2 \right). 
	\end{align*}
	Now, (\ref{eq:Harnack}) implies that
	\begin{equation*}
		\int_Q u^2 \leq \vert Q \vert \, \ess \, u^2 \leq  \vert Q \vert \, \cH^2 \, \underset{Q/2}{\operatorname{ess} \inf} \, u^2 
		\leq  \cH^2 \int_{Q/2} u^2
		\leq \cH^2 \int_{\mathbb{R}^d} \chi_Q^2u^2,
	\end{equation*}
	where the last inequality follows from the properties of $\chi_Q$. This yields the claim we had set to prove.
\end{proof}

Together, Lemmas~\ref{lma:UpperBound} and \ref{lma:LowerBound} yield the claim of Theorem~\ref{thm:main}. Finally we prove Corollary \ref{cor:confining}.

\begin{proof}[Proof of Corollary~\ref{cor:confining}]
	If $u$ vanishes at infinity, i.e. $\limsup_{R\rightarrow \infty} \sup_{\mathbb{R}^d\setminus B(0,R)} u =0$, then each sublevel set of $1/u$ is bounded up to a null set and %have compact closure and 
	thus $H$ has discrete spectrum by \eqref{eq:MainBathtub}. Assume on the other hand that $u$ does not vanish at infinity. There is $\mu>0$ and a sequence of points $(x_n)_{n\geq 1}$ such that $\lim_{n\rightarrow \infty} \vert x_n \vert=\infty$ and $\liminf_{\varepsilon\rightarrow 0^+} \inf_{B(x_n,\varepsilon)} \, u\geq \frac{\cH}{\mu}$ for all $n$. Then by \eqref{eq:Harnack} we have
	\begin{equation*}
		\bigcup_{n\geq 1} Q\left(x_n, 2\sqrt{\cH/\mu}\right) \subseteq \{ x\in \mathbb{R}^d \ : \ 1/u(x) \leq \mu \}
	\end{equation*}
	and hence, by \eqref{eq:MainBathtub}, the spectrum of $H$ is not discrete.
\end{proof}		

		\section{The case of polynomial potentials} \label{sect:cor}

When the potential $V$ is a polynomial, as in the original setting of Fefferman-Phong, one can obtain more precise information of the landscape function. We start by giving the proof for Corollary \ref{cor:Polynomials}.

\begin{proof}[Proof of Corollary~\ref{cor:Polynomials}]
	Since the addition of a constant does not change the structure of the spectrum, we assume that the polynomial satisfies $V\geq 1$. We check first that these polynomials satisfy \eqref{assump:Kato} and \eqref{assump:doubling}.
	Condition \eqref{assump:Kato} holds with $\delta=2$ due to the inequality
	\begin{align*}
		c \sup_{B(x,r)} V \leq \frac{1}{\vert B(x,r)\vert} \int_{B(x,r)} V(y) dy \leq \sup_{B(x,r)} V,
	\end{align*}
	where $c$ can be chosen to depend only on $d$ and the total degree of $V$, but neither $x$ nor $r$. The upper bound is immediate. It is enough to show the lower bound for $r=1$ and $x=0$ by scaling and translation. In that case, the claim follows from the fact that the space of all polynomials in $d$ variables and total degree at most $D$ is a finite dimensional vector space and thus all norms are equivalent.
	
	For the same reason and since polynomials are analytic functions, there exists a constant $C>0$ depending only on $d$ and $D$ such that 
	\begin{align*}
		\int_{B(0,2)}  V(y)  dy \leq C \int_{B(0,1)}  V(y)  dy,
	\end{align*} 
	which implies doubling after rescaling and translation. In particular, $\cD$ can be chosen to only depend on $d$ and $D$.
	
	Now, Corollary~\ref{cor:confining} and~(\ref{eq:equivum}) imply that the spectrum of $H$ is discrete if and only if $\lim_{\vert x\vert \rightarrow \infty} m(x,V)=\infty$. For polynomials the Fefferman-Phong-Shen maximal function $m(x,V)$ is in fact equivalent to $M(x,V)$ introduced in~(\ref{SmithZhongMaxFunction}), in the sense that 
	\begin{equation} \label{eq:Mequivm}
		cM(x,V) \leq m(x,V) \leq C M(x,V).
	\end{equation}		
	The equivalence was already noted in~\cite{She2} and we provide a proof below for completeness, see Lemma~\ref{lma:M-m}.
	
With these preliminaries, we can now turn to the central claim of the corollary. If one of the directional derivative vanishes, then $-\Delta+V$ is unitarily equivalent (via a suitable rotation) to $-\Delta+W$ where $\partial_1 W \equiv 0$. In this case $M((t,0, \dots, 0),W)=M(0, W)$, which implies by the remarks above that the spectrum of $-\Delta+W$ is not discrete and hence also the spectrum of $-\Delta+V$ is not discrete.

Next we are going to show that if $-\Delta+V$ does not have discrete spectrum, then some directional derivative of $V$ vanishes identically. As $-\Delta + V$ does not have discrete spectrum we must have that
\begin{equation} \label{M0}
	\liminf_{\vert x \vert \rightarrow \infty} M(x,V) =: M_0 <\infty.
\end{equation}
We consider the semi-algebraic set
\begin{align*}
	A= \left\{ x\in \mathbb{R}^d \ \vert \ \forall \alpha \in \mathbb{N}^d \ : \ (\partial^\alpha V (x))^2 < 2 M_0^{2(2+\vert \alpha \vert)} \right\}
\end{align*}
and the polynomial function
\begin{align*}
	F: \mathbb{R}^d \rightarrow \mathbb{R}^{(D+1)^d}, \quad x\mapsto (\partial^\alpha V(x)^2 )_{\alpha\in [0,D]^d \cap \mathbb{Z}^d}.
\end{align*}
Now \eqref{M0} implies that $A$ is an unbounded set and we can therefore pick a sequence $(x^{(n)})_{n\in \mathbb{N}}\subseteq A$ such that $\vert x^{(n)}\vert \rightarrow \infty$ and
\begin{align*}
	\lim_{n \rightarrow \infty} F(x^{(n)}) =: y =(y_\alpha)_{\alpha\in [0,D]^d \cap \mathbb{Z}^d}
\end{align*} 
with $\vert y_\alpha \vert \leq M_0^{2(2+\vert \alpha \vert)}$.

Next we would like to pass from a mere sequence to an analytic curve. This is done by the following curve selection lemma at infinity.
\begin{lemma}{\cite[Lemma 2.17]{HN}}
	Let $A\subset \mathbb{R}^d$ be a semi-algebraic set, and let $F: \mathbb{R}^d \rightarrow \mathbb{R}^N$ be a semi-algebraic map. Assume that there exists a sequence $(x^{(n)})_{n\in \mathbb{N}} \subset A$ such that $\lim_{n\rightarrow \infty} \vert x^{(n)} \vert = \infty$ and $\lim_{n \rightarrow \infty} F(x^{(n)})=y\in (\mathbb{R}\cup \{\pm \infty\})^N$. Then there exists an analytic curve $\gamma: (0, \delta)\rightarrow A$ of the form
	\begin{equation} \label{gamma}
		\gamma(t) = \sum_{j=-m}^\infty a^{(j)} t^j
	\end{equation} 
	such that $a^{(-m)}\in \mathbb{R}^N \setminus\{0\}$, $m\in \mathbb{Z}_{>0}$ and $\lim_{t\rightarrow 0^+} F(\gamma(t))=y$.
\end{lemma}
Let $\gamma$ be a curve as given by the previous lemma. We would like to say that $V$ remains constant along $\gamma$ and thus get a direction in which the gradient of $V$ vanishes identically. However, analytic functions can remain bounded on an unbounded set without being constant. Thus, we truncate the series \eqref{gamma} at $j=0$, thereby obtaining a polynomial approximation of the curve $\gamma$, and $F$ will still remain bounded along the truncation.
\begin{lemma}\label{lemma:perturbation of polynomial}
	For every $\varepsilon>0$ there exists $C>0$ such that for all $v\in \mathbb{R}^d$ with $\vert v \vert<\varepsilon$ we have for all $t\in (0, \delta/2)$
	\begin{equation}
		0\leq V(\gamma(t)+v) \leq C.
	\end{equation}
\end{lemma}
\begin{proof}
	By Taylor's theorem, $V(\gamma(t)+v) = \sum_{\alpha\in \mathbb{N}^d} \frac{v^\alpha}{\alpha!} (\partial^\alpha V)(\gamma(t))$.	The claim follows from the fact that $\vert(\partial^\alpha V)(\gamma(t))\vert$ are all uniformly bounded for $t\in (0,\delta/2)$.
\end{proof}
With this, we define polynomial function 
\begin{align*}
	G(s) = \sum_{j=0}^m a^{(-j)} s^j.
\end{align*}
Note that every component of $G$ is single variable polynomial. For every $\varepsilon>0$ there exists $0<\delta_\varepsilon<\delta$ such that
\begin{align*}
	\vert \gamma(t) - G(t^{-1}) \vert <\frac{\varepsilon}{2}
\end{align*}
for all $t\in (0,\delta_\varepsilon)$, and hence by the previous lemma
\begin{align*}
	0\leq V(G(t^{-1})) \leq C.
\end{align*}
Let now
\begin{align*}
	P(s,x_1, \dots, x_d) = V\left( G(s) + x \right).
\end{align*}
For $s>\frac{1}{\delta_\varepsilon}$ and $\vert x \vert<\frac{\varepsilon}{2}$ we get
\begin{align*}
	0\leq P(s, x_1, \dots, x_d) \leq C
\end{align*}
again by Lemma~\ref{lemma:perturbation of polynomial}. Now, for any $\vert x \vert < \frac{\varepsilon}{2}$, the function $s\mapsto P(s,x)$ is a polynomial that is bounded on an unbounded interval and thus constant. Therefore, $\partial_s P(s,x)=0$ on $\bbR\times B_{\varepsilon/2}(0)$. By the identity theorem we get that $\partial_s P(s,x)=0$ on $\mathbb{R}^{d+1}$. But
\begin{align*}
	0=\partial_s P(s,x) = (\nabla V)(G(s)+x) \cdot G'(s).
\end{align*}
As $G$ is not constant, there is $s_0\in \mathbb{R}$ such that $G'(s_0)\neq 0$ and so the derivative of $V$ in direction $G'(s_0)$ vanishes identically.
\end{proof}

	As mentioned before, the landscape function will not belong to the form domain of $H$. For polynomial potentials, there is an easy criterion to check whether $u\in \mathrm{dom}(H^{1/2})$.
	
	\begin{lemma} \label{lm:formdomain}
		Let $V\geq 0$ be a non-zero polynomial. Then $u\in L^1(\mathbb{R}^d)$ if and only if $u\in \mathrm{dom}(H^{1/2})$.
	\end{lemma}
	\begin{proof}
		As $V$ is smooth, we get by standard elliptic regularity theory that the landscape function is a classical solution of the landscape equation. Multiplying the landscape equation by $\varphi\in C_c^\infty(\mathbb{R}^d)$ and integrating yields after integration by parts
		\begin{align*} %\label{eq:Caccioppoli}
			\int_{\mathbb{R}^d} \varphi \vert \nabla u \vert^2 + \int_{\mathbb{R}^d} Vu^2 \varphi
			= \int_{\mathbb{R}^d} u \varphi + \int_{\mathbb{R}^d} u^2 (\Delta \varphi).
		\end{align*}
		We saw in the proof of Corollary \ref{cor:Polynomials} that either $\lim_{\vert x \vert \rightarrow \infty} u(x) =0$ or that there is one spatial coordinate along which $u$ is constant. Hence, if $u\in L^1(\mathbb{R}^d)$, then $u$ vanishes at infinity and automatically $u\in L^2(\mathbb{R}^d)$. However, then we can choose a sequence of $\varphi_n\in C_c^\infty(\mathbb{R}^d)$ converging to $1$ and obtain by monotone convergence
		\begin{align*}
			\int_{\mathbb{R}^d}\left( \vert \nabla u \vert^2 + Vu^2\right)
			\leq  \int_{\mathbb{R}^d} u + C\int_{\mathbb{R}^d} u^2 <\infty
		\end{align*}
		and therefore $u\in \mathrm{dom}(H^{1/2})$.
		
		On the other hand, if $u\in \mathrm{dom}(H^{1/2})$, then $u\in H^1(\mathbb{R}^d)$ and in particular $u\in L^2(\mathbb{R}^d)$. Thus, we can take again a suitable sequence of test functions to obtain by dominated convergence
		\begin{align*}
			\int_{\mathbb{R}^d} u = \int_{\mathbb{R}^d} \vert \nabla u \vert^2 + \int_{\mathbb{R}^d} Vu^2 - \int_{\mathbb{R}^d} u^2 <\infty,
		\end{align*}
		namely $u\in L^1(\bbR^d)$.
	\end{proof}
		
Let us now return to the concrete example of Simon's potential $V(x,y)=x^2 y^2$. First of all, we can now prove that the corresponding landscape function is in the form domain of $-\Delta + V$. Combining Proposition \ref{Prop:Equivalence} and Lemma \ref{lma:M-m}, it is enough to check that $M(\cdot,x^2y^2)^{-2} \in L^1(\mathbb{R}^2)$. An explicit calculation yields $M((x,y), x^2y^2)\geq \vert x y \vert + \sqrt{\vert x\vert} + \sqrt{\vert y \vert}+1$, and thus its inverse is indeed square integrable in $\bbR^2$. Similarly, we obtain the following two-sided estimate on the effective potential:
\begin{equation*}
	c \left(x^2 y^2 + \vert x \vert + \vert y \vert+1  \right) \leq \frac{1}{u(x,y)} \leq C\left(x^2 y^2 + \vert x \vert + \vert y \vert+1  \right).
\end{equation*}
This yields for $\mu$ sufficiently large
\begin{align*}
	\frac{\mu}{3C} \log\left(\frac{\mu}{3C}\right) = \left\vert \left\{ (x,y) \in \mathbb{R}^2 \ : \ 1\leq x \leq \frac{\mu}{3C}, 0\leq y \leq \frac{\mu}{3Cx} \right\} \right\vert
	\leq \caV(\mu).
\end{align*}
On the other hand, we have
\begin{align*}
	\caV(\mu) &\leq 1 + 4 \left\vert \left\{ (x,y)\in \left[1,\frac{\mu}{c}\right]\times \mathbb{R}_{\geq 0} \ : \ x^2y^2+\vert x \vert +\vert y\vert+1 \leq \frac{\mu}{c} \right\} \right\vert \\
	&\leq 1 + 4 \left\vert \left\{ (x,y)\in \left[1,\frac{\mu}{c}\right]\times \mathbb{R}_{\geq 0} \ : \ \vert xy \vert \leq \frac{\mu}{c} \right\} \right\vert \\
	&= 1+ \frac{4\mu}{c} \log\left( \frac{\mu}{c} \right).
\end{align*}
The combination of those two estimates with Theorem~\ref{thm:main} recovers Simon's asymptotics \cite[Theorem 1.4]{Si2}
\begin{equation*}
	\caN^{x^2 y^2}(\mu) = \frac{1}{\pi} \mu^{3/2} \log(\mu) + o(\mu^{3/2} \log(\mu)).
\end{equation*}
up to multiplicative constants.

We conclude this section with a proof of the equivalence of the functions $m(\cdot,V)$ and $M(\cdot,V)$ in the case of polynomials. We point out that the arguments in this section show that in the case of polynomials, the constants appearing in Theorem~\ref{thm:main}, and \emph{a fortiori} the Harnack constant, depend only on the spatial dimension and the degree of the polynomial.
\begin{lemma}\label{lma:M-m}
	Let $V\geq 0$ a polynomial on $\mathbb{R}^d$ of total degree $D\geq 0$. Then there exist constants $C,c>0$ depending only on $d,D$ such that 
	\begin{equation} \label{eq:Mequivm}
		cM(x,V) \leq m(x,V) \leq C M(x,V),
	\end{equation}
	where $m(\cdot,V), M(\cdot,V)$ were defined in \eqref{eq:defm(x,V)} and~(\ref{SmithZhongMaxFunction}).
\end{lemma}
\begin{proof} By translating the potential, we can pick $x=0$. Furthermore, for all $\lambda>0$ we have $m(x, V_\lambda) =  \lambda m(\lambda x,V)$ and $M(x, V_\lambda) =  \lambda M(\lambda x,V)$, where $V_\lambda(x) = \lambda^2 V(\lambda x)$. Hence, we can assume that $M(0,V)=1$ and then need to show that there exists $c>0$ depending only on $d,D$ such that $m(0,V)\geq c$.
	
	Since $V\mapsto \partial^\alpha V$ is a linear map on a finite dimensional space
	\begin{equation*}
		\vert \partial^\alpha V(0) \vert \leq \sup_{x\in B(0,1)} \vert \partial^\alpha V(x) \vert \leq c \sup_{x\in B(0,1)} \vert V(x) \vert,
	\end{equation*} 
	and so
	\begin{align*}
		1 = M(0,V) &\leq \sum_{\alpha \in \mathbb{N}_0^n} \Big(c\sup_{x\in B(0,1)} V(x)\Big)^{1/(\vert \alpha \vert+2)} \\
		&\leq c \Big( \Big(\sup_{x\in B(0,1)} V(x)\Big)^{1/2} + \Big(\sup_{x\in B(0,1)} V(x)\Big)^{1/(D+2)} \Big).
	\end{align*}
	Hence, we have $\sup_{x\in B(0,1)} V(x) \geq c>0$. For $r\geq 1$ we get
	\begin{align*}
		c r^2 \leq r^2 \sup_{x\in B(0,r)} V(x) \leq \frac{c}{r^{n-2}} \int_{B(0,r)} V(y) dy.
	\end{align*}
	Recall that $\cD$ can be chosen to only depend on $d,D$, therefore the right hand side is greater than $\cD$ for $r$ large enough, which yields an upper bound on $\frac{1}{m(0,V)}$. Hence, 
	\begin{align*}
		m(0,V) \geq c = c M(0,V).
	\end{align*}
	
	We turn to the lower bound. First of all, a simple Taylor expansion yields (see~\cite[Lemma 2.5]{Sm})
	\begin{equation*}
		\vert \partial^\alpha V(y) \vert\leq C M(x,V)^{\vert \alpha \vert+2} \left( 1+ \vert x -y \vert M(x,V) \right)^D,
	\end{equation*}
	so that
	\begin{equation} \label{eq:m}
		M(y,V) \leq C M(x,V)(1+\vert x -y \vert M(x,V))^{D/2},
	\end{equation}
	for all $x,y\in \mathbb{R}^d$. Thus, if $m(0,V)=1$ then
	\begin{align*}
		\cD^{-(2+D)} m(0,V)^{2+D} &= \int_{B(0,1)} V(y)dy \leq  \vert B(0,1)\vert \sup_{y\in B(0,1)} V(y)\\
		& \leq \vert B(0,1)\vert \sup_{y\in B(0,1)} M(y,V)^2 \leq C^2 M(0,V)^{2+D},
	\end{align*}
	by~(\ref{eq:m}), which again yields the desired estimate by translating and rescaling.
\end{proof}

\section{The potential well}
In this section we explicitly compute the landscape function for potential wells $\varepsilon \mathbbm{1}_{B(0,\delta)^c}$, where $B(0,\delta)^c=\mathbb{R}^d\setminus B(0,\delta)$ and $\varepsilon, \delta>0$. We shall observe first of all that the minimum of the effective potential properly reflects the value of the bottom of the spectrum in the sense that both are of order $\varepsilon$ as $\varepsilon\to0$. Secondly, we will see that the estimates of the main theorem are not tight enough to distinguish the difference between $d=1,2$, where an eigenvalue is present for all $\varepsilon>0$, and $d\geq3$ where this is not the case.

We start by observing that the landscape function corresponding to the spherical well are radially symmetric. Indeed, all the Lax-Milgram solutions \eqref{eq:uL} are invariant under rotation of the first $d$ variables and thus the landscape function, given as a pointwise limit of those solutions, shares the same symmetry. Passing to spherical coordinates we see that the radial part $f(\vert x \vert) = u(x)$ solves the ODE
\begin{align*}
	-f''(r)-\frac{d-1}{r} f'(r) + \varepsilon \mathbbm{1}_{[\delta, \infty)}(r) f(r) =1
\end{align*}
on $(0,\infty)$. The general solution of this ODE on $(0,\delta)$ is given, for $d\neq 2$, by
\begin{align*}
	f(r) = -\frac{r^2}{2d} + a_1 + \frac{a_2}{r^{d-2}},
\end{align*}
respectively by the same expression with $r^{-(d-2)}$ replaced by $\log(r)$ for $d=2$.
As $\lim_{r\rightarrow 0^+} f(r) = \lim_{r\rightarrow 0^+} u(re_1)=u(0)$, we conclude in the case $d\geq 2$ that $a_2=0$. The same follows for $d=1$ as $u$ is even and $C^1(\bbR)$.

On the other hand, on $(\delta,\infty)$, the general solution is given by
\begin{align*}
	f(r) = \frac{1}{\varepsilon} + b_1 r^{1-\frac{d}{2}} K_{-1+d/2}(\sqrt{\varepsilon}r) + b_2 r^{-1+\frac{d}{2}} I_{-1+d/2}(\sqrt{\varepsilon}r),
\end{align*}
where $I_m, K_m$ denote the modified Bessel function of the first, respectively the second kind. We have $\lim_{r\rightarrow \infty} I_m(r)=\infty$ and $\lim_{r\rightarrow \infty} K_m(r)=0$ for $m\geq -1/2$ (use \cite[Eqs. 9.6.10, 9.6.23]{AS} and $K_{-1/2}(x)=\sqrt{\pi}e^{-x}/\sqrt{2x}, I_{-1/2}(x)=\sqrt{2}\cosh(x)/\sqrt{\pi x}$).
As $0$ is not in the spectrum of $-\Delta +\varepsilon \mathbbm{1}_{B(0,\delta)^c}$, we get that $u$ is bounded and hence $b_2=0$. This yields
\begin{align*}
	f(r) = \begin{cases}
		-\frac{r^2}{2d} + a_1,& r\in (0,\delta),\\
		\frac{1}{\varepsilon} + b_1 r^{1-\frac{d}{2}} K_{-1+d/2}(\sqrt{\varepsilon}r),& r\in (\delta, \infty).
	\end{cases}
\end{align*}
Finally, the coefficients can be determined by the fact that $f\in C^1(\mathbb{R}_{>0})$.
In dimensions $d=1,3$ the Bessel functions can be expressed in elementary functions and the solutions are given by
\begin{align*}
	u(x) = \begin{cases}
		-\frac{|x|^2}{2} + \frac{1}{\varepsilon} + \frac{\delta}{\sqrt{\varepsilon}} + \frac{\delta^2}{2},& \vert x \vert\leq \delta,\\
		\frac{1}{\varepsilon} + \frac{\delta}{\sqrt{\varepsilon}} e^{-\sqrt{\varepsilon}(\vert x \vert-\delta)},& \vert x \vert>\delta
	\end{cases}
\end{align*}
for $d=1$ and by
\begin{align}\label{1dWell}
	u(x) = \begin{cases}
		- \frac{\vert x \vert^2}{6} + \frac{1}{\varepsilon}+\delta^2\left( \frac{1}{6} + \frac{1}{1+\sqrt{\varepsilon}\delta} \right) ,& \vert x \vert \leq \delta,\\
		\frac{1}{\varepsilon} + \frac{\delta^3}{1+\sqrt{\varepsilon} \delta} \frac{e^{-\sqrt{\varepsilon}(\vert x \vert-\delta)}}{\vert x \vert},& \vert x \vert>\delta
	\end{cases}
\end{align}
for $d=3$.

In all dimensions we have that $u$ is radially symmetric and its radial part is monotone decreasing (even exponentially). Furthermore, we have $\lim_{\vert x\vert \rightarrow \infty} u(x) = \frac{1}{\varepsilon}$. This implies that the sublevel set $ \caV(\mu)$ of the effective potential $1/u$ is monotone increasing, remains finite for $\mu<\varepsilon$ and $\lim_{\mu \rightarrow \varepsilon^-}\caV(\mu)=\infty$. This is consistent with the fact that the bottom of the essential spectrum is $\epsilon$ and $c<1,C>1$ in  Theorem~\ref{thm:main}.

For $d=1$ the smallest eigenvalue $\mu_0$ of $-\Delta + \varepsilon \mathbbm{1}_{B(0,\delta)^c}$, for $0<\varepsilon$ sufficiently small, is the smallest positive solution of 
	\begin{align*}
		\sqrt{\varepsilon - \mu_0} = \sqrt{\mu_0} \tan(\sqrt{\mu_0} \delta).
	\end{align*}
	Thus, for $\delta>0$ fixed, we obtain $\mu_0=\varepsilon (1-O(\sqrt{\varepsilon}))$ as $\varepsilon\to 0^+$. 
As discussed at the beginning of the section, this is the same asymptotic behaviour as that of the minimum of $1/u$, see~(\ref{1dWell}). The same holds for $d=3$ where however the bottom of the spectrum is the bottom of the essential spectrum, namely $\caN(\mu)=0$ for $\mu<\varepsilon$ and $\caN(\mu)=\infty$ for $\mu\geq \varepsilon$. Here, $\caV(\mu)$ is arbitrarily large for $\mu\to\varepsilon^-$, showing that $c<1$ in \eqref{eq:MainBathtub}.

\subsection*{Acknowledgements}
The three authors acknowledge financial support from NSERC of Canada.


\begin{thebibliography}{55}
			
			\bibitem{AS} Abramowitz, M., \& Stegun, I. A. (1966). Handbook of mathematical functions with formulas, graphs, and mathematical tables. Washington: US Govt. Print.
			
			\bibitem{ADFM} Arnold, D. N., David, G., Filoche, M., Jerison, D., \& Mayboroda, S. (2019). Localization of eigenfunctions via an effective potential. Communications in Partial Differential Equations, 44(11), 1186-1216.
			
			\bibitem{AFMWZ} Arnold, D., Filoche, M., Mayboroda, S., Wang, W., \& Zhang, S. (2022). The landscape law for tight binding Hamiltonians. Communications in Mathematical Physics, 396, 1339–1391.
			
			\bibitem{BLG} Balasubramanian, S., Liao, Y., \& Galitski, V. (2020). Many-body localization landscape. Physical Review B, 101(1), 014201.
			
			\bibitem{Ca} Calogero, F. (1965). Upper and lower limits for the number of bound states in a given central potential. Communications in Mathematical Physics, 1, 80–88.
			
			\bibitem{CWZ} Chenn, I., Wang, W., \& Zhang, S. (2022). Approximating the ground state eigenvalue via the effective potential. Nonlinearity, 35(6), 3004.
			
			\bibitem{Ch} Christ, M. (1991). On the $\overline{\partial}$ equation in weighted $L^2$ norms in $\mathbb{C}^1$. The Journal of Geometric Analysis, 1(3), 193-230.
			
			\bibitem{Cw} Cwikel, M. (1977). Weak type estimates for singular values and the number of bound states of Schrödinger operators. Annals of Mathematics, 106, 93-100.
			
			\bibitem{DHM} Davey, B., Hill, J., \& Mayboroda, S. (2016). Fundamental matrices and Green matrices for non-homogeneous elliptic systems. Publicacions Matemàtiques. 62(2), 537-614.
			
			\bibitem{DFM} David, G., Filoche, M., \& Mayboroda, S. (2021). The landscape law for the integrated density of states. Advances in Mathematics, 390, 107946.
			
			\bibitem{Dav} Davies, E. B. (1995). Spectral theory and differential operators (Vol. 42). Cambridge University Press.
		
			
			
			\bibitem{FP} Fefferman, C. L. (1983). The uncertainty principle. Bulletin of the American Mathematical Society 9.2: 129-206.
			
			\bibitem{FMT} Filoche, M., Mayboroda, S., \& Tao, T. (2021). The effective potential of an M-matrix. Journal of Mathematical Physics, 62(4), 041902.
			
			\bibitem{F} Frank, R. L. (2014). Cwikel's theorem and the CLR inequality. Journal of Spectral Theory, 4(1), 1-21.
			
			\bibitem{FLW} Frank, R. L., Laptev, A., \& Weidl, T. (2022). Schrödinger operators: eigenvalues and Lieb–Thirring inequalities (Vol. 200). Cambridge University Press.
			
			\bibitem{GN} Grigor’yan, A., \& Nadirashvili, N. (2015) Negative Eigenvalues of Two-Dimensional Schrödinger Operators. Archive for Rational Mechanics and Analysis 217, 975–1028.
		

			\bibitem{HN} Ha, H. V., \& Nguyen, T. T. (2022). Łojasiewicz gradient inequalities for polynomial functions and some applications. Journal of Mathematical Analysis and Applications, 509(1), 125950.
			
			\bibitem{HHRV} Hoang, V., Hundertmark, D., Richter, J., \& Vugalter, S. (2023). Quantitative bounds versus existence of weakly coupled bound states for Schrödinger type operators. In Annales Henri Poincaré, 24(3), 783-842.
			
			\bibitem{HKRV} Hundertmark, D., Kunstmann, P., Ried, T., \& Vugalter, S. (2023). Cwikel’s bound reloaded. Inventiones Mathematicae, 231(1), 111-167.
			
			\bibitem{Iv} Ivrii, V. (2016). 100 years of Weyl’s law. Bulletin of Mathematical Sciences, 6(3), 379-452.
			
			
			\bibitem{La1} Laptev, A. (2000). The Negative Spectrum of a Class of Two-Dimensional Schrödinger Operators with Potentials Depending Only on Radius. Functional Analysis and Its Applications, 34, 305-307.
			
			\bibitem{NL} Laptev, A., \& Netrusov, Y. (1999). On the negative eigenvalues of a class of Schrödinger operators. American Mathematical Society Translations, 189, 173 - 186.
			
			\bibitem{LRS} Laptev, A., Read, L., \& Schimmer, L. (2022). Calogero type bounds in two dimensions. Archive for Rational Mechanics and Analysis, 245(3), 1491-1505.
			
			\bibitem{Li} Lieb, E. H. (1976). Bounds on the eigenvalues of the Laplace and Schrödinger operators. Bulletin of the American Mathematical Society, 82(5), 751–753.
		
			
			\bibitem{LL} Lieb, E. H., \& Loss, M. (2001). Analysis (Vol. 14). American Mathematical Society.
			
			\bibitem{LT} Lieb, E. H., \& Thirring, W. E. (1976). Inequalities for the moments of the eigenvalues of the Schrödinger Hamiltonian and their relation to Sobolev inequalities. Pages 269–303 of: Studies in Mathematical Physics (Essays in Honor of Valentine Bargmann).
			E. H. Lieb, B. Simon, and A. S. Wightman (editors). Princeton University Press,
			Princeton, NJ.
			
			\bibitem{LS} Liskevich, V., \& Skrypnik, I. I. (2010). Harnack inequality and continuity of solutions to elliptic equations with nonstandard growth conditions and lower order terms. Annali di Matematica Pura ed Applicata, 189, 335-356.
			
			\bibitem{MN} Mohamed, A., \& Nourrigat, J. (1991). Encadrement du $N(\lambda)$ pour un opérator de Schrödinger
			avec un champ magnétique et un potentiel électrique, J. Math. Pures Appl. 70,
			87–99.
			
			\bibitem{MLBN} Mohamed, A., Lévy-Bruhl, P., \& Nourrigat, J. (1993). Etude spectrale d'opérateurs liés à des représentations de groupes nilpotents. Journal of functional analysis, 113(1), 65-93.
			
			\bibitem{O1} Otelbaev, M. (1976). Bounds for eigenvalues of singular differential operators. Mathematical notes of the Academy of Sciences of the USSR, 20(6), 1038-1042.
			
			\bibitem{O2} Otelbaev, M. (1979). Imbedding theorems for spaces with a weight and their application to the study of the spectrum of a Schrödinger operator, Trudy Mat. Inst. Steklov., (Procedings of the Steklov Mathematical Institute), 150, 265-305.
			
			\bibitem{PW} Payne, L. E., \& Weinberger, H. F. (1960). An optimal Poincaré inequality for convex domains. Archive for Rational Mechanics and Analysis, 5(1), 286-292.

			
			\bibitem{P} Poggi, B. (2021). Applications of the landscape function for Schr\" odinger operators with singular potentials and irregular magnetic fields. arXiv preprint arXiv:2107.14103.
			
			\bibitem{MP} Poggi, B. G., \& Mayboroda, S. (2021). Exponential decay of fundamental solutions to Schrödinger operators and the landscape function. In 2021 Spring Western Virtual Sectional Meeting. AMS.
			
			
			\bibitem{RS4} Reed, M., \& Simon, B. (1978) Methods of Modern Mathematical Physics IV: Analysis of Operators. Academic Press, Inc.
			
			\bibitem{Rob} Robert, D. (1982). Comportement asymptotique des valeurs propres d’opérateurs du type
			Schrödinger à potentiel “dégénéré”, J. Math. Pures Appl. 61(9), 275–300.
			
			\bibitem{Ro} Rozenblyum, G. V. (1976). Distribution of the discrete spectrum of singular differential operators. Izvestiya Vysshikh Uchebnykh Zavedenii. Matematika, (1), 75-86.
			
			\bibitem{Ro1} Rozenblyum, G. V. (1978). Estimates of the spectrum of the Schrödinger operator. Journal of Soviet Mathematics, 10(6), 934-944.
			
			\bibitem{She1} Shen, Z. (1995). $ L^ p $ estimates for Schrödinger operators with certain potentials. Annales de l'institut Fourier, 45(2), 513-546.
			
			\bibitem{She2} Shen, Z. (1996). Eigenvalue asymptotics and exponential decay of eigenfunctions for Schrödinger operators with magnetic fields. Transactions of the American Mathematical Society, 348(11), 4465-4488.
			
			\bibitem{She3} Shen, Z. (1998) On bounds of $N(\lambda)$ for a magnetic Schrödinger operator, Duke Mathematical Journal, 94(3), 479-507.
			
			\bibitem{She4} Shen, Z. (1999). On fundamental solutions of generalized Schrödinger operators. Journal of Functional Analysis, 167(2), 521-564.
			
			\bibitem{Sha} Shargorodsky, E. (2014). On negative eigenvalues of two-dimensional Schrödinger operators. Proceedings of the London Mathematical Society, 108(2), 441-483.
			
			\bibitem{Si1} Simon, B. (1983). Some quantum operators with discrete spectrum but classically continuous spectrum. Annals of Physics, 146(1), 209-220.
			
			\bibitem{Si2} Simon, B. (1983). Nonclassical eigenvalue asymptotics. Journal of functional analysis, 53(1), 84-98.
		
			
			\bibitem{Sm} Smith, H. F. (1991). Parametrix Construction for a Class of Subelliptic Differential Operators, Duke Mathematical Journal, 63(2), 343-354.
			
			\bibitem{WZ} Wang, W., \& Zhang, S. (2021). The exponential decay of eigenfunctions for tight-binding Hamiltonians via landscape and dual landscape functions. Annales Henri Poincaré, 22, 1429-1457.
			
			\bibitem{St1} Steinerberger, S. (2017). Localization of quantum states and landscape functions. Proceedings of the American Mathematical Society, 145(7), 2895-2907. 
			
			\bibitem{St2} Steinerberger, S. (2021). Regularized potentials of Schrödinger operators and a local landscape function. Communications in Partial Differential Equations, 46(7), 1262-1279.
			
			\bibitem{W1} Weyl, H. Über die asymptotische Verteilung der Eigenwerte. Nachrichten von der Gesellschaft der Wissenschaften zu Göttingen, Mathematisch-Physikalische Klasse 1911
			
			
			\bibitem{W2}  Weyl, H. (1912). Das asymptotische Verteilungsgesetz der Eigenwerte linearer partieller Differentialgleichungen (mit einer Anwendung auf die Theorie der Hohlraumstrahlung). Mathematische Annalen, 71(4), 441-479.
			
			
			\bibitem{W3} Weyl, H. (1912). Über die Abhängigkeit der Eigenschwingungen einer Membran und deren Begrenzung. Journal für die reine und angewandte Mathematik 141: 1-11.
			
			\bibitem{W4} Weyl, H. (1913). Über die Randwertaufgabe der Strahlungstheorie und asymptotische Spektralgesetze. Journal für die reine und angewandte Mathematik 143: 177-202.
			
			\bibitem{Zh} Zhong, J. (1993) Harmonic Analysis for Some Schr\"odinger Type Operators, Ph.D. Thesis,
			Princeton University.
		\end{thebibliography}
	\end{document}